\documentclass{article}

\usepackage{amsmath,amsthm,amssymb,enumerate,bbold,fullpage}
\usepackage[ocgcolorlinks,colorlinks=true,urlcolor=blue,linkcolor=blue]{hyperref}
\usepackage[pdftex]{graphicx}
\usepackage[utf8]{inputenc}

\newtheorem{theorem}{Theorem}
\newtheorem{lemma}[theorem]{Lemma}
\newtheorem{corollary}[theorem]{Corollary}

\newtheorem{problem}[theorem]{Open problem}
\newtheorem{conjecture}[theorem]{Conjecture}

\newcommand{\R}{{\mathbb{R}}}
\newcommand{\N}{{\mathbb{N}}}
\newcommand{\G}{\ensuremath{\mathcal{G}}}
\newcommand{\OO}{\ensuremath{\mathcal{O}}}

\newcommand{\pth}[1]{\ensuremath{\left(#1\right)}}
\newcommand{\pbrcx}[1]{\ensuremath{\left[#1\right]}}
\newcommand{\ceil}[1]{\ensuremath{\left\lceil#1\right\rceil}}

\newcommand{\Prob}[1]{\mathbb{P}\pbrcx{#1}}
\newcommand{\Ex}[1]{\mathbb{E}\pbrcx{#1}}

\DeclareMathOperator{\area}{area}

\newcommand{\ljp}{57}

\title{On Order Types of Random Point Sets\thanks{Funded by grant
    ANR-17-CE40-0017 of the French National Research Agency (ANR
    project ASPAG). This work was initiated during the ALEA 2013
    conference and the 15th INRIA–McGill–Victoria Workshop on
    Computational Geometry at the Bellairs Research Institute.}}

\author{
  Olivier Devillers\thanks{Université  de Lorraine, CNRS, Inria, LORIA, F-54000 Nancy, France. \texttt{Olivier.Devillers@inria.fr}}
  \and Philippe Duchon\thanks{LaBRI, Université de Bordeaux, CNRS, Bordeaux INP, F-33504 Talence, France. \texttt{philippe.duchon@u-bordeaux.fr}}
  \and Marc Glisse\thanks{Inria, Centre de recherche Saclay-Île-de-France, France. \texttt{Marc.Glisse@inria.fr}}
  \and Xavier Goaoc\thanks{Université  de Lorraine, CNRS, Inria, LORIA, F-54000 Nancy, France. Partially funded by Institut Universitaire de France. \texttt{xavier.goaoc@loria.fr}}}

\begin{document}
\maketitle

\begin{abstract}
A simple method to produce a random order type is to take the order
type of a random point set. We conjecture that many probability
distributions on order types defined in this way are heavily
concentrated and therefore sample inefficiently the space of order
types. We present two results on this question. First, we study
experimentally the bias in the order types of $n$ random points chosen
uniformly and independently in a square, for $n$ up to $16$. Second,
we study algorithms for determining the order type of a point set in
terms of the number of coordinate bits they require to know. We give
an algorithm that requires on average $4n \log_2 n+O(n)$ bits to
determine the order type of $P$, and show that any algorithm requires
at least $4n \log_2 n - O(n \log\log n)$ bits. This implies that the
concentration conjecture cannot be proven by an ``efficient encoding''
argument.
\end{abstract}

\setcounter{page}{0}

\clearpage

\section{Introduction}

An order type is a combinatorial abstraction of a finite point
configuration that already determines which subsets are in convex
position and which pairs define intersecting segments. (Hence, the
order type of a point set $P$ encodes the convex hull, the convex
pealing structure, the triangulations of~$P$, or, for instance, which
graphs admit straight-line embeddings with vertices mapped to~$P$.) In
this paper, we study the problem of producing random order types
efficiently and with limited bias.

\subsection{Context}\label{s:context}

The \emph{orientation} of a triple $(a,b,c) \in (\R^2)^3$ is the sign
of the determinant
\[ \left| \begin{matrix} a_x & b_x & c_x
  \\ a_y & b_y & c_y \\ 1 & 1 & 1 \end{matrix}\right|, \qquad \text{where $a_x$ is the $x$-coordinate of $a$, etc.}\]
This sign is $-1$ if the triangle $abc$ is oriented clockwise (CW),
$0$ if it is flat, and $1$ if it is oriented counterclockwise
(CCW). Two sequences $P = (p_1,p_2,\ldots, p_n) \in (\R^2)^n$ and $Q =
(q_1,q_2, \ldots, q_n) \in (\R^2)^n$ \emph{have the same chirotope} if
for every indices $i,j,k$ the triples $(p_i,p_j,p_k)$ and
$(q_i,q_j,q_k)$ have the same orientation. A related notion is for two
finite subsets $P$ and $Q$ of $\R^2$ to \emph{have the same order
type}, meaning that there exists a bijection $f: P \to Q$ that
preserves orientations. Having the same order type (resp. chirotope)
is an equivalence relation, and an
\emph{order type} (resp. a \emph{chirotope}) is an equivalence class for
that relation. An order type or chirotope is \emph{simple} if it can
be realized without three collinear points. These definitions extend
readily to $\R^d$, but we consider here only the planar, simple case.

\paragraph{Order types VS chirotopes.}

The questions we are interested in are usually oblivious to the
labeling of the points, and are therefore phrased in terms of order
types. Our methods do, however, make explicit use of the labeling of
the points, so our results are stated in terms of chirotopes for the
sake of precision. We therefore use one or the other notion depending
on the context. They are related since an order type of size $n$
corresponds to at most $n!$ chirotopes, possibly fewer if some
bijections of the point set into itself preserve orientations.

\paragraph{Enumerating order types.}

There are finitely many order types of size $n$, so, in principle,
some properties of planar point sets of small size can be studied by
sheer enumeration of order types.\footnote{Here is an example, coming
from geometric Ramsey theory, of such a ``constant size'' open
question. Gerken~\cite{gerken2008empty} proved that any set of at
least $1717$ points in the plane without aligned triple contains an
\emph{empty hexagon}: six points in convex position with no other
point of the set in their convex hull. The largest known point set
with no empty hexagon has size $29$ and was found decades
ago~\cite{overmars2002finding}.} In practice, order types were
enumerated (up to possible reflexive symmetry) up to size $11$ by
Aichholzer \textit{et al.}~\cite{aichholzer2002enumerating}. They used
their database for instance to establish sharp bounds on the minimum
and maximum numbers of triangulations on $10$ points, a very finite
result that they could bootstrap into an asymptotic bound. The number
of order types of size $n$ does, however, quickly become overwhelming
as $n$ increases: it reaches billions already for $n=11$, and grows at
least as $n^{3n+o(n)}$ since the number of chirotopes grows as
$n^{4n+\Theta\pth{\frac{n}{\log
n}}}$~\cite[Theorem~4.1]{alon1986number}.  It is thus unlikely that
the order type database will be extended much beyond size
$11$.\footnote{In particular, the geometric Ramsey theory problem
above seems out of reach of enumerative methods.}

\subsection{Questions}

When dealing with configuration spaces too large to be enumerated, it
is natural to fall back on random sampling methods. Two desirable
properties of a random generator of order types are that it be 
efficient (a random order type can be produced quickly, say in time
polynomial in $n$) and reasonably unbiased (it will explore a
reasonably large fraction of the space of order types).

\paragraph{Challenges.}

Designing an efficient and reasonably unbiased random generator of order
types may prove difficult because of two properties of order types. On
the one hand, order types enjoy small combinatorial encodings, even of
subquadratic size~\cite{cardinal_et_al:LIPIcs:2018:8733}, but the set
of order types is difficult to describe: already deciding membership
is NP-hard~\cite{shor1991stretchability}. On the other hand, order
types can be manipulated through point sets realizing them, so that
one needs not worry about remaining in the space of order types, but
there are order types of size $n$ for which any realization requires
$2^{\Omega(n)}$ bits per coordinate~\cite{goodman1990intrinsic}.

\paragraph{Concentration.}

Let us illustrate what we consider \emph{un}reasonable bias. Let $m_n$
be a sequence of positive integers with $m_n \to \infty$, and let
$\mu_n$ be a probability measure on the set of order types of size
$m_n$. Say that $\{\mu_n\}_{n \in \N}$ exhibits \emph{concentration}
if there exists for each $n$ a set $S_n$ of order types of size $m_n$
such that $S_n$ contains a proportion $\epsilon_n\to 0$ of all order
types, while $\mu_n(S_n)\to 1$. In other words, $\mu_n$ and the
counting measure on order types of size $m_n$ are ``asymptotically
singular''. In fact, little seems known already on the following
question.
 
\begin{problem}\label{op}
  Does there exist a sequence of measures $\mu_n$ on order types of
  size $m_n$ such that (i) no subsequence exhibits
  concentration, and (ii) a random order type of size $m_n$ according to
  measure $\mu_n$ can be produced in time polynomial in $n$?
\end{problem}

\paragraph{Random point sets.}

It is easy to produce a random order type by first generating a random
point set, then reading off its order type, but let us stress that it
is not clear how the probability distribution on point sets translates
into a probability distribution on order types.

\begin{problem}\label{op2}
   How biased is the order type of a set of points sampled from a
   planar measure (say uniform on a square)?
\end{problem}

\noindent
When sampling points independently and from a probability distribution
whose support has non-empty interior, every order type appears with
positive probability.  Indeed, every order type can be realized on an
integer grid~\cite{goodman1990intrinsic} and order types are unchanged
under rescaling and sufficiently small perturbation. One may
still expect some bias, if only because some order types
require exponential precision for their
realization~\cite{goodman1990intrinsic} and are thus more brittle
than others. For order types of small size, bias was proven to be
unavoidable~\cite[Prop. 2]{goaoc2015limits}.

\subsection{Results}

For the sake of clarity, we state and prove our results for
a \emph{uniform sample of a square}, understood as a
\emph{sequence} of random points chosen independently and uniformly in
$[0,1]^2$ (the choice of which square does not affect the
distribution). We comment in Section~\ref{s:conclusion} to what extent
our methods generalize. We write $\log$ to mean the logarithm of base
$2$.

\paragraph{Experiments.}

Our first contribution (Section~\ref{s:exp}) is some experimental
evidence that sampling random point sets uniformly and independently
in a square explores \emph{very inefficiently} the space of order
types for $n$ up to $16$. This prompts:

\begin{conjecture}\label{c:conc}
  Let $\mu_n$ denote the probability distribution on order types of
  size $n$ given by uniformly sampling a square. The sequence
  $\{\mu_n\}_{n\in \N}$ exhibits concentration.
\end{conjecture}

\paragraph{Algorithms.}

Recall that the number of chirotopes grows as
$n^{4n+\Theta\pth{\frac{n}{\log n}}}$. An ``entropic'' approach to
proving (the chirotopal analogue of) Conjecture~\ref{c:conc} could
thus be to find an algorithm that reads off the chirotope of a uniform
sample of a square using with high probability at most $c n \log n$
random bits, for some $c<4$. Formally, we consider a discrete model of
computation (e.g., a Turing machine), \emph{not} the real-RAM machine
customary in computational geometry, where reading the coordinates has
a cost (specifically, accessing the next bit in one of these strings
has unit cost) and any other computation is considered free. A random
point set is then given in the form of $2n$ infinite binary strings,
one per point coordinate\footnote{Recall that any real $r \in [0,1]$
has a binary development of the form $0.r_1r_2\ldots$ with
$r_i \in \{0,1\}$, so we can identify $r$ with the sequence
$r_1r_2\ldots \in \{0,1\}^\N$. (In particular, the real $1$ is
identified with the sequence $1^{\N}$; for dyadic reals, which have
two representations, we can choose any.)} and we want to determine its
chirotope efficiently most of the time. Our second contribution
establishes that such an approach fails:

\begin{theorem}\label{th:unif}
  Let $P$ be a uniform sample of size $n$ in $[0,1]^2$.

  \begin{itemize}
  \item[(i)] Any algorithm that determines the chirotope of $P$ reads
    on average at least $4n \log n-O(n\log\log n)$ coordinate bits.

    \medskip
    
  \item[(ii)] There exists an algorithm that determines the chirotope
    of $P$ by reading on average $4n \log n+ O(n)$ coordinate bits.
  \end{itemize}
\end{theorem}

\noindent
We prove Theorem~\ref{th:unif} in two steps. First,
Section~\ref{s:setup} answers the questions listed above for
an \emph{arbitrary} point set $P$ in terms of two statistics ($L$ and
$U$) of that point set. Sections~\ref{s:U} and~\ref{s:L} then make a
probabilistic analysis of these statistics for our random point sets.

\paragraph{Another approach.}

Our proof of Theorem~\ref{th:unif}~(ii) uses a similar argument
(namely Lemma~\ref{l:dU}) as the following result of Fabila-Monroy and
Huemer~\cite{fabila2017order}: with probability at least
$1-O\pth{n^{-\epsilon}}$, a uniform sample of a square of size
$n$ can be rounded to the regular grid of step $n^{-3-\epsilon}$
without changing its chirotope. Can most chirotopes be realized on a
$O\pth{n^{3+\epsilon}}\times O\pth{n^{3+\epsilon}}$ regular grid? A
negative answer would prove Conjecture~\ref{c:conc}. Unfortunately,
the best bounds that we are aware of, due to Caraballo et
al.~\cite{caraballo2018number}, do not settle this question: they only
assert that the number of chirotopes of resolution $n^{-3-\epsilon}$
is at least $n^{3n- O(n \log \log n/\log n)}$, whereas the number of
chirotopes is $n^{4n+\Theta\pth{\frac{n}{\log n}}}$.

\section{Experimental study of order types of random point sets}
\label{s:exp}

In this section, we probe experimentally the probability distribution
of order types of uniform samples of a square. Note that the number of
order types is about $28$ million for size 10, between $2.3$ and $4.7$
billion for size 11 (see Appendix~\ref{a:exp}), and unknown for $n \ge
12$.

\paragraph{Setup.}

Our first experiment is to produce a large number $N$ of point sets,
stopping after each million samples to record the empirical
distribution of order types. We repeated this experiment 80 times for
size $10$ (for $N=1$ billion) and 20 times for size $11$ (for $N=450$
million). For size $12$, we ran out of memory before getting useful
information (we used machines with 16 to 64 gigabytes of memory.)

It seems plausible that the expected number of samples needed to reach
a repetition provides some insight on how concentrated that measure
is; for comparison, this expectation is $\Theta(\frac1{\sqrt{k}})$ for a
uniform measure on $k$ elements. We thus set up a second experiment
where we produce point sets until we reach the first repetition of an
order type. We repeated this experiment 10000 times for each
size from $10$ to $14$, 5468 times for size $15$ and 1000 times for
size $16$.

Due to lack of space, we defer the discussion of technical issues to
Appendix~\ref{a:exp}.

\paragraph{Data.}

We present here synthetic views of our experimental results.

\begin{figure}[!h]
\includegraphics[page=11]{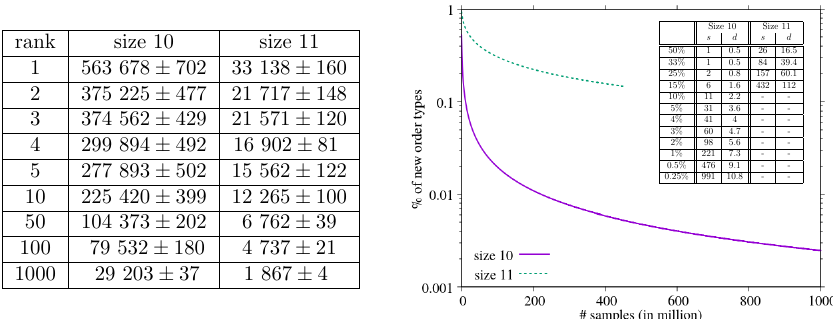}
\caption{Results for the 1st experiment, averaged over 80 trials for size 10 and 20 trials for size 11. Left: Number of occurrences (average $\pm$ standard deviation) for the $k$th most frequent order types for $k=1, 2, \ldots, 1000$. Right: The proportion of new order types found (the scale is logarithmic on the $y$-axis). The table gives some triples $(x\%, s, d)$, meaning that after $s$ million samples, $d$ million distinct order types were found and $x\%$ of the last million samples were new ones.
\label{f:res}}
\end{figure}

\vspace{-0.3cm}
\begin{table}[!h]
\begin{center}
\begin{tabular}[b]{|c|c|c|c|c|c|c|c|}
\hline
Size & 10 & 11 & 12 & 13 & 14 & 15 & 16 \\
\hline
Average & 466& 2\ 716& 18\ 788& 156\ 372& 1\ 521\ 365& 17\ 134\ 843& 218\ 060\ 427\\
\hline
Median & 432 & 2\ 546& 17\ 540& 147\ 266& 1\ 429\ 508& 16\ 027\ 384& 203\ 340\ 042\\
\hline
\end{tabular}
\end{center}
\caption{Results for the 2nd experiment: time of first collision
(averaged over 10\ 000 trials for size 10 to 14, 5468 for size 15 and
1000 for size 16).\label{t:res2}}
\end{table}

\paragraph{Discussion.}

For size 10 and 11, the empirical frequencies of the most popular
order types are $5.6\ 10^{-4}$ and $7.3 \ 10^{-5}$, which are several
orders of magnitude above the corresponding uniform probability ($3.5\
10^{-8}$ and about $4\ 10^{-10}$, respectively). This behavior
persists, as even the 1000th popular order type remains $3$ to $4$
orders of magnitude more frequent than for the uniform
behavior. Notice that (Figure~\ref{f:res} right) the rate at which new
order types are discovered collapses quickly: for size 10, after
seeing $\sim 2.2\ 10^6$ distinct order types, in the next million
samples only $10\%$ produce a new order type; this means that
$\sim7.7\%$ of the order types of size 10 account for $90\%$ of the
mass. The situation seems similar for size 11. Altogether, this
suggests that uniform samples of a square explore very inefficiently
the space of order types.

This first assessment may seem weakly justified as it is based on mere
averages. We do not provide a statistical analysis of these
estimators, but note that the random variable counting the number of
distinct order types seen after $t$ samples is a sum of $t$ Bernoulli
variables that are not independent, but are \emph{negatively
associated} in the sense explained in Section~\ref{s:Lbb}. This
variable therefore enjoys a Chernoff-type tail estimate, and can be
accurately estimated through averaging over a reasonably small number
of samples. This is consistent with the rather small standard
deviations observed on our samples. We thus believe that these
empirical averages represent the situation quite fairly.

\bigskip

Our second experiment indicates that for size $12$ to $16$ the time of
first collision remain orders of magnitude smaller than what it would be
for a uniform distribution. Indeed, let us write $T_n$ for the number
of order types of size $n$, and speculate on the value of $\sqrt{T_n}$
($T_n$ is unknown for $n \ge 12$). For $n=10$ this is $5348$; the
ratio $\frac{T_{n+1}}{T_{n}}$ increases regularly as $n$ ranges from
$5$ to $10$, and is about $160$ for $n=10$. Assuming that it does not
decrease, the value $\sqrt{T_n}$ should grow by a factor at least
$12$, and likely much more. The average empirical first collision time
grows by smaller factors: $5.8$, $6.9$, $8.3$, $9.7$, $11.2$ and $12.9$.

\bigskip

Let us sketch a more refined analysis. Let $p_{n,1}, p_{n,2}, \ldots,
p_{n,T_n}$ denote the probabilities of the various order types in a
random sample of a square, in non-increasing order. Let $\mu_n$ denote
the probability distribution on $[T_n]$ such that $\Prob{\mu_n=i} =
p_{n,i}$. Let $C_{\mu_n}$ denote the time of first collision for
$\mu_n$ and let us identify $\mu_n$ to the vector $(p_{n,1},
p_{n,2}, \ldots, p_{n,T_n})$. Camarri and
Pitman~\cite[Corollary~5]{CamPit00} proved that
$\frac{C_{\mu_n}}{\|\mu_n\|_2}$ asymptotically follows the Rayleigh
distribution with density $x \exp(-x^2/2)$ if and only if
$\|\mu_i\|_{\infty} = o\pth{\|\mu_i\|_2}$. We found (by hand) a
scaling of the times of first collision obtained experimentally so
that their distribution seems to fit

\noindent
\begin{minipage}{7cm}
 that Rayleigh distribution (see figure on the right); a
Kolmogorov-Smirnov test confirms that for $n=11$ to $16$, these
normalized data are consistent with such a convergence. The hypothesis
is asymptotic in nature and we only sampled order types up to size 16,
but given that there are already billions of order types for size 11,
this provides some (weak) evidence in favor of $\|\mu_i\|_{\infty} =
o\pth{\|\mu_i\|_2}$. Assuming this indeed holds, $\sqrt{2/\pi}
C_{\mu_n}$ is an asymptotically unbiased estimator for $\|\mu_i\|_2$
(the constant factor comes from the mean of the Rayleigh
distribution).
\end{minipage}
\hfill
\begin{minipage}{7cm}
\includegraphics[page=12]{Figures}
\end{minipage}
\hfill

\bigskip

Does $\|\mu_i\|_{\infty} = o\pth{\|\mu_i\|_2}$ relate to
concentration? On one hand, it is compatible with $\mu_i$ being
uniform on $i$ elements. On the other hand, if $\mu_i$ charges
uniformly $n_i = o(i)$ elements for a total of $1-\frac1i$ and
uniformly $i-n_i$ elements for a total of $\frac1i$, the condition is
only satisfied for $n_i = \omega\pth{\sqrt{i}}$. Perhaps this
condition prevents too sharp a concentration.

\section{Analysis of arbitrary point sets}
\label{s:setup}

We first introduce one algorithm and two statistics to analyze the
information needed to determine the order type of an \emph{arbitrary}
set $P$ of $n$ points in the unit square, no three aligned.

\paragraph{Grids and orientations.}

Let $G_m$ denote the partition of $[0,1]^2$ into $m\times m$ square
cells of side length $\frac1m$ where the interior of each cell is of
the form $\pth{\frac im,\frac{i+1}m} \times\pth{\frac jm,\frac{j+1}m}$
with $0 \le i,j < m$. We often set $m=2^k$, so that knowing the first
$k$ bits of both coordinates of a point amounts to knowing which cell
of $G_m$ contains it. Remark that knowing three points up to $k$ bits
(for each coordinate) suffices to determine their orientation if and
only if the corresponding three cells of $G_m$ cannot be intersected
by a line.

\paragraph{Greedy algorithm.}

The algorithm that we propose for Theorem~\ref{th:unif}~(ii) refines
greedily the coordinates of a point involved in a triangle with
undetermined orientation, until the chirotope can be determined.  We
start with no bit read, so we only know that all points are in the
unit square. At every step, we select one point and read one more bit
for both of its coordinates. So, at every step of the algorithm, we
know for each point some grid cell that contains it; the resolution of
the grid may of course be different for every point. The selection is
done greedily as follows:

\begin{quote}
  \emph{Find three pairwise distinct indices $a,b,c$ such that the
    cells known to contain $p_a$, $p_b$, $p_c$ can be intersected by
    a line, and select one among these points known to the coarsest
    resolution.}
\end{quote}

\noindent
We break ties arbitrarily, so this is perhaps a method rather than an
algorithm. By definition, when the algorithm stops, the chirotope of
$P$ can be determined from the precision at which every point is
known. The algorithm does \emph{not} stop if $P$ contains three
aligned points.

\paragraph{Statistic $U$.}

For $i \in [n]$, we define $U(i)$ as the smallest $k$ such that for
any $a,b \in [n]\setminus \{i\}$, there does not exist a line that
intersects the cells in $G_{2^k}$ that contain $p_a$, $p_b$, and
$p_i$. If $p_i$ is aligned with some two other points of $P$, we let
$U(i) = \infty$. The following implies that our greedy algorithm
terminates if $P$ has no aligned triple.

\begin{lemma}\label{l:U}
  In the greedy algorithm above, independently of how ties are
  resolved, for every $i \in [n]$, at most $U(i)$ bits are read from
  each coordinate of $p_i$.
\end{lemma}
\begin{proof}
  Assume that at some point in the algorithm, we read the $k$th bit of
  both coordinates of point $p_i$. To read these bits, our selection
  method requires that there exist $a,b \in [n]\setminus \{i\}$ such
  that (1) in $\{p_a,p_b,p_i\}$, $p_i$ is one of the points known at
  coarsest resolution, and (2) there exists a line intersecting the
  cells known to contain $p_a$, $p_b$, and $p_i$.  Condition~(1)
  ensures that for each of $\{p_a,p_b,p_i\}$, the cell known to
  contain the point is contained in a cell of
  $G_{2^{k-1}}$. Condition~(2) ensures that these cells in
  $G_{2^{k-1}}$ can be intersected by a line. Thus, $k-1 < U(i)$.
\end{proof}

\paragraph{Statistic $L$.}

For $i \in [n]$, we define $L(i)$ as the smallest $k$ such
that \emph{at least one} horizontal or vertical segments of length
$2^{-k}$ starting in $p_i$ is \emph{disjoint} from \emph{all} lines
$p_ap_b$ with $a,b\in [n]\setminus\{i\}$.

\begin{lemma}\label{l:L}
  Any algorithm that determines the chirotope of $P$ must read, for
  every $i$, at least $L(i)-1$ bits of each coordinate of $p_i$.
\end{lemma}
\begin{proof}
  Assume that we know $k$ bits of the $x$-coordinate of the
    point $p_i$. The set of possible positions for $p_i$ then contains
    a horizontal segment $S$ of length $2^{-k}$ containing $p_i$; in
    fact, it would be exactly such a segment if we knew the
    $y$-coordinate of $p_i$ to infinite precision.

  By definition of $L(i)$, the two horizontal segments of
    length $2^{-(L(i)-1)}$ starting in $p_i$ both intersect some line
    $p_ap_b$ with $a,b \in [n]\setminus \{i\}$ (the lines are
    different for the two segments). If $2^{-k} \ge 2\cdot
    2^{-(L(i)-1)}$, then the segment $S$ contains at least one of
    these horizontal segments, and is also intersected by some line
    $p_ap_b$ with $a,b \in [n]\setminus \{i\}$. Since the possible
    positions of $p_i$ contain $S$, this means that the bits read so
    far from $p_i$ do not suffice to determine the orientation of
    the triple $(p_i, p_a, p_b)$, even if $p_a$ and $p_b$ were known to
    infinite precision.

    Conversely, if an algorithm that determines the chirotope of $P$
    reads $k$ bits from the $x$-coordinate of $p_i$, then we must have
    $2^{-k} < 2\cdot 2^{-(L(i)-1)}$, that is $k > L(i)-2$. The same
    argument applies to the $y$-coordinate of $p_i$.
\end{proof}

\paragraph{From here...}

So the minimal number of bits required\footnote{Note that our lower
bound holds for \emph{any} algorithm that determines the chirotope,
provided it reads the bits of each coordinate in order, starting from
the most significant. It is in particular not assumed that the
algorithm always reads as many bits of the two coordinates for a given
point, although our proposed algorithm does respect this condition.}
to determine the chirotope of $P$ is in between $2\sum_{i=1}^n
(L(i)-1)$ and $2 \sum_{i=1}^n U(i)$. We show in the next sections that
both sums equal, at first order, $4n \log n$ on average when $P$ is a
uniform sample of the unit square.

\section{Probabilistic analysis of $U$}\label{s:U}

We now outline an analysis of the random variable $U(\cdot)$ when $P$
is a uniform sample of the unit square. The variables $U(i)$ have the
same distribution, which satisfies

\newcounter{postponedtheorem:ExU}\setcounter{postponedtheorem:ExU}{\value{theorem}}
\begin{lemma}\label{l:ExU}
  $\Ex{U(1)} \le 2\log n + 8$
\end{lemma}

\noindent\begin{minipage}{0.5\textwidth}
Given two points $p$ and $q$ in $[0,1]^2$, the \emph{butterfly}
$B_m(p,q)$ is the set of positions of a point $r \in [0,1]^2$ such
that the cells of $G_m$ containing these three points do not determine
the orientation of $(p,q,r)$. Formally, $B_m(p,q)$ is the union of all
cells of $G_m$ that intersect a line secant to the cells of $G_m$ that
contain $p$ and~$q$. The random variable $U(i)$ equals the smallest
$k$ such that $\bigcup_{j\neq i}B_{2^k}(p_i,p_j)$ contains no other
point of $P$. We prove Lemma~\ref{l:ExU} by bounding from above the
area of a butterfly in terms of $m$ and the distance $pq$ and applying
a union  bound. Fabila-Monroy and Huemer\cite{fabila2017order} introduced a
very  close notion of butterfly (\emph{c.f.} their sets $F_{i,j}$) to
study how rounding coordinates affects order types. They already performed the analysis we
\end{minipage}
\quad
\begin{minipage}{0.4\textwidth}
\includegraphics[page=1,width=\textwidth]{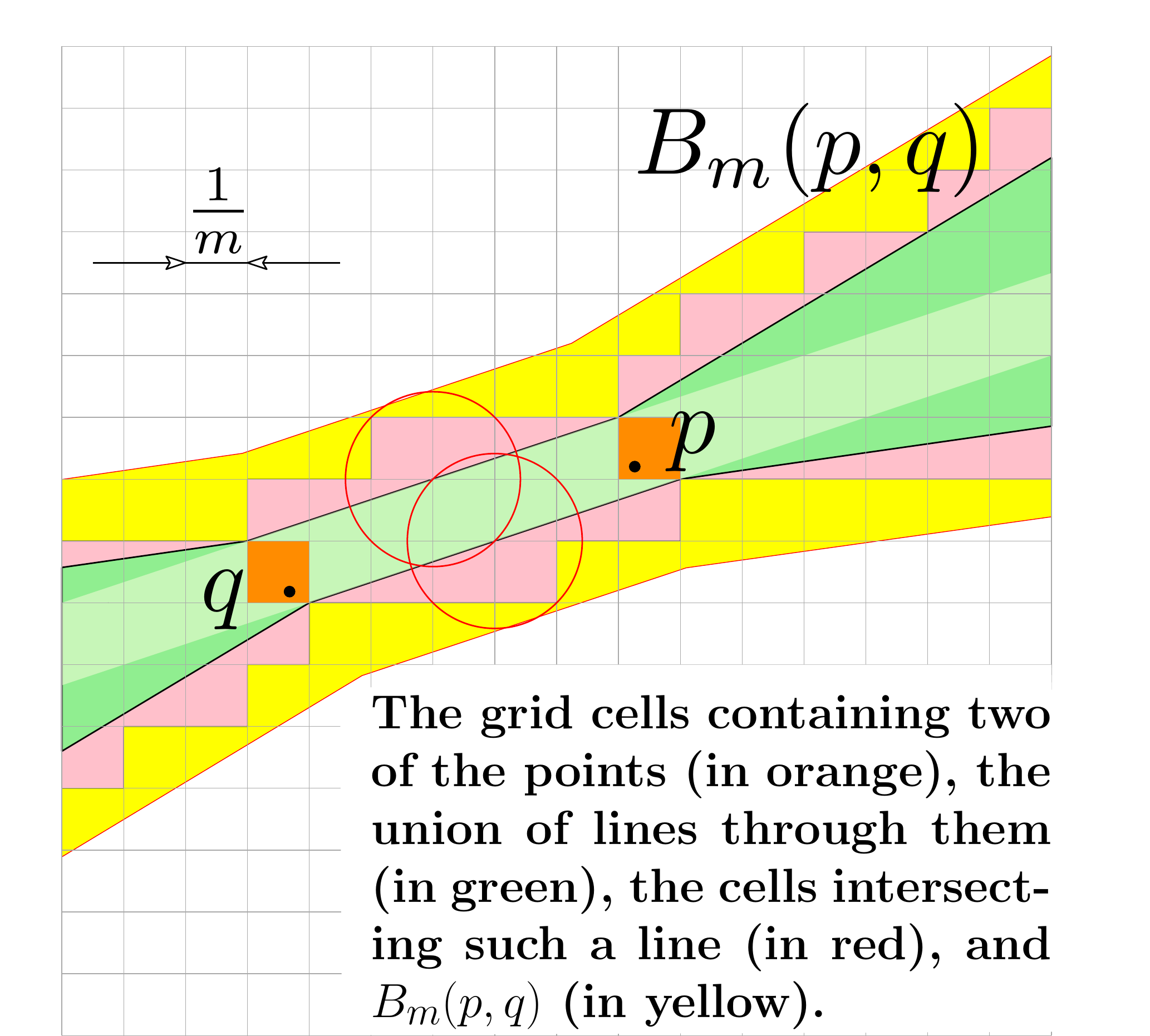}
\end{minipage}\smallskip 

\noindent
need~\cite[Lemmas~3 and~4]{fabila2017order}, so we only spell out the
proof of Lemma~\ref{l:ExU} in Appendix~\ref{a:U} for completeness.

\bigskip

We can now prove that our greedy algorithm for deciding the chirotope
of~$P$ reads on average at most $4n \log n + O(n)$ coordinate bits.

\begin{proof}[Proof of Theorem~\ref{th:unif}~(ii)]
  By Lemma~\ref{l:U}, our greedy algorithm reads at most $U(a)$ bits
  from each coordinate of point $p_a$. Thus, using Lemma~\ref{l:ExU},
  the average number of bits used by our algorithm is at most:
  $  2\Ex{\sum_{i=1}^n U(i)} = 2 \sum_{i=1}^n \Ex{U(i)}= 2n \Ex{U(1)} \leq 4n\log n + 16 n.  $
\end{proof}

\section{Probabilistic analysis of $L$}
\label{s:L}

We now outline an analysis of the random variable $L(\cdot)$ when $P$
is a uniform sample of the unit square. Again, the variables $L(i)$
have the same distribution. The key technical result is:

\begin{lemma}\label{l:tailL}
  For every $x>1$, there exists $c>0$ such that $\Prob{L(1) \ge 2 \log
    n - x\log\log n}$ is at least $1-2^{-cn}$.
\end{lemma}

\noindent
Proving Lemma~\ref{l:tailL} will take the rest of the section, but let us start by using it. 

\begin{proof}[Proof of Theorem~\ref{th:unif}~(i)]
  All $n$ variables $L(i)$ have the same expectation and by
  Lemma~\ref{l:L}, any algorithm that determines the chirotope of $P$
  must read at least a total of $2 (\sum_i L(i)-1)$. It thus suffice
  to determine $\Ex{L(1)}$, which rewrites as $\Ex{L(1)} = \sum_{k\geq
  0} \Prob{L(1)>k}$.

  Note that $\Prob{L(1)>k}$ decreases with $k$. Lemma~\ref{l:tailL}
  for $x=\frac32$ implies that the first $2 \log n - \frac32\log \log
  n$ terms are at least $1-2^{-cn}$ for some constant $c>0$. Keeping
  only these terms, we get
  \[ \Ex{L(1)} \geq (1-2^{-cn}) \pth{2 \log n - \frac32\log \log
  n} \geq (1-2^{-cn}) 2\log n - \frac32\log \log n.\]
  For $n$ large enough, $2^{-cn+1} \log n < \frac12\log \log n$ and
  the statement follows.
\end{proof}

\subsection{Discretization into a bichromatic birthday problem}

Our approach is to look for lines passing close to $p_1$, as such
lines are likely to force $L(1)$ to be large. To do so, we divide the
plane into some number of angular sectors around $p_1$ and define a
blue disk of center $p_1$ and radius $0.2$ and a red annulus with
center $p_1$ and radii $0.3$ and $0.4$. This discretizes the problem,
as if we find two points of $P$ in the blue and red parts of the same
or nearby sectors, then they must span a line passing close to $p_1$.

\bigskip
\noindent\begin{minipage}{0.6\textwidth}
 \quad One technical issue is that if $p_1$ is close enough to the
boundary of $[0,1]^2$, then parts of the red and blue regions will be
outside of $[0,1]^2$ and cannot contain any point of $P$. We handle
this by considering the 4 diagonal directions $(\pm 1,\pm 1)$, and
picking the one in which the boundary is the furthest away from
$p_1$. Now, in the cone of half-angle $\pi/8$ around that direction,
the red and blue parts are contained in the unit square. Letting $8s$
denote the total number of sectors, we therefore focus on the $s$
sectors around that direction. For the rest of this section, we assume
that this direction is $(1,1)$ as illustrated by the figure; the three
other cases are symmetric. We label $B_1, B_2, \ldots B_{s}$
(resp. $R_1, R_2, \ldots R_{s}$) the intersection of each of our
angular sectors with the blue disk minus $p_1$ (resp. the red
annulus), in counterclockwise order.
\end{minipage}\hfill\begin{minipage}{0.37\textwidth}
 \includegraphics[page=4,width=0.8\textwidth]{Figures}
\end{minipage}\hfill

\smallskip\noindent

\bigskip

Next, finding a line close to $p_1$ is not enough: to ensure that
$L(1) > k$, we need to find lines that intersect \emph{all four}
horizontal and vertical segments of length $2^{-k}$ with endpoint
$p_1$. To do that, we look for lines $(br)$ where $b \in B_i$ and $r
\in R_{i+1}$. This shift in indices ensures that the line $(br)$ is
close to $p_1$ \emph{and} passes below $p_1$: indeed, $r$ and $b$ are
respectively above and below the ray from $p_1$ that is a common
boundary of $B_i$ and $R_{i+1}$. Similarly, finding some points $b'\in
B_{i'}$ and $r' \in R_{i'-1}$ will provide a line $(b'r')$ passing
close to $p_1$ and above it; together, these two lines will intersect
all four horizontal and vertical segments that have $p_1$ as an
endpoint. It remains to relate the size of these segments to $s$.

\bigskip

Since we consider what happens around the direction $(1,1)$, the line
passing below $p_1$ will have to intersect both the horizontal segment
with $p_1$ as leftmost point and the vertical segment with $p_1$ as
topmost point. Note, however, that any line $(br)$ that we consider
has slope at least $\tan \frac{\pi}8 = \sqrt{2}-1$. Let
$\rightarrow_\ell$ and $\downarrow_\ell$ denote the segments of length
$\ell$ with $p_1$ as, respectively, leftmost and topmost point. If a
line $(br)$ intersects $\downarrow_{\ell}$, it must also intersect
$\rightarrow_{\frac\ell{\sqrt{2}-1}}$. We thus focus on finding the
smallest $\ell$ such that $\downarrow_\ell$ is guaranteed to meet
$(br)$.

\newcounter{postponedtheorem:cross}\setcounter{postponedtheorem:cross}{\value{theorem}}
\begin{lemma}\label{l:cross}
  If $s\ge 10$, for any point $b \in B_i$ and $r \in R_{i+1}$, the line
  $(br)$ intersects $\downarrow_{\frac{\pi}{2s}}$.
\end{lemma}
\begin{proof}
The proof involves elementary geometry and is deferred to Appendix~\ref{a:L}.
\end{proof}

Altogether, we can bound $L(1)$ from below by a simple balls-in-bins condition:

\begin{corollary}\label{c:anniv}
  Assume that $k \ge 3$ and that $s=2^{k+1}$. If there exists $i,i'$ in
  $[s]$ such that $P$ intersects each of the four regions $B_i$, $R_{i+1}$,
  $B_{i'}$, and $R_{i'-1}$, then $L(1) \ge k$.
\end{corollary}
\begin{proof}
  Let $b\in B_i \cap P$ and $r\in R_{i+1} \cap P$. Since $s \ge
  16$, Lemma~\ref{l:cross} ensures that the line $(br)$
  intersects $\downarrow_{\frac{\pi}{2s}}$. As argued before
  Lemma~\ref{l:cross}, that line also intersects $\downarrow_{\ell}$
  and $\rightarrow_{\ell}$ with $\ell =
  \frac{\pi}{2(\sqrt{2}-1)s}$. A symmetry with respect to
  the line of slope 1 through $p_1$ gives the intersection with the two
  other segments from the points in $B_{i'}$ and $R_{i'-1}$. Since
  $\frac{\pi}{2(\sqrt{2}-1)} \le 4$, the existence of $i$
  and $i'$ ensures that all four horizontal and vertical segments of
  length $\frac{4}s = 2^{-k+1}$ starting in $p_i$ are intersected by some lines
  spanned by $P\setminus\{p_1\}$, so $L(1) > k-1$.
\end{proof}

\subsection{A balls-in-bins analysis}
\label{s:Lbb}

To prove Lemma~\ref{l:tailL}, we are interested in the probability
that $L(1)$ be at least $2n \log n$ (minus some change), so we use
Corollary~\ref{c:anniv} with $s = \frac{n^2}{\log^x n}$ and $x >1$. To
study the probability that the indices $i$ and $i'$ exist, we define,
for $i \in [s]$ and $j \in [n-1]$, the random variables
\begin{center}
  \begin{tabular}{l|l|l}
    $X_{i,j} = \mathbb{1}_{p_{j+1}\in B_i}$ & $X_i=\max_{j \in [n-1]} X_{i,j}$ & $X = \sum_{i \in [s]}X_i$\\
    $Y_{i,j} = \mathbb{1}_{p_{j+1}\in R_i}$ & $Y_i=\max_{j \in [n-1]} Y_{i,j}$ & $Y = \sum_{i \in [s]}Y_i$ \\
  \end{tabular}
\end{center}
(Note that, for a better bookkeeping, we index the events associated
with $p_j$ by $j-1$ because $p_1$ is already chosen.) In plain
English, $X_i$ is the indicator variable that $B_i$ is non-empty, and
$X$ counts the number of non-empty regions $B_i$. (The $Y_i$ and $Y$
variables do the same for the regions $R_i$.)  The definition of the
regions ensures that each is fully contained in the unit square, that
all $B_i$ have the same area, and that all $R_i$ have the same
area. So all the $\{X_{i,j}\}_{i,j}$ are identically distributed, and
so are the $\{Y_{i,j}\}_{i,j}$, the $\{X_i\}_i$, and the $\{Y_i\}_i$.

\paragraph{Approach.}

Conditioning on $X=\beta$ and $Y=\rho$, there are $\beta$ or $\beta-1$
red cells whose index follows the index of an occupied blue cell
(depending on $B_s$).  Since the $\rho$ occupied red cells are chosen
uniformly amongst the $s$ red cells, we get:
\begin{equation}\label{eq:simple}
  \Prob{\exists i \colon B_i \cap P \neq \emptyset \text{ and } R_{i+1} \cap P \neq \emptyset \mid X=\beta, Y = \rho} \ge 1-\pth{1-\frac{\beta-1}{s}}^\rho.
\end{equation}
Indeed, all but at most one of the occupied blue cells are next to a
red cell which, if occupied, makes the event true. Our approach is to
combine this inequality with a concentration bound for $X$ and $Y$ to
bound from below the probability that $i$ exists. A symmetric argument
takes care of the existence of $i'$.

\paragraph{Concentration of sums of dependent variables.}

If the $X_i$ and the $Y_i$ were independent, the Chernoff-Hoeffding
would bound from below the values of $X$ and $Y$ with high
probability. For fixed $j$, however, any subset of
$\{X_{i,j}\}_i \cup \{Y_{i,j}\}_i$ sums to zero or one; These variables
are thus ``negatively'' dependent in the sense that when one is $1$,
the others must be $0$. Formally, they can be shown to
be \emph{negatively associated}. We do not elaborate on this notion
here, but refer to the paper of Dubhashi and Ranjan~\cite{DubRan98}
from which we highlight the following points:

\begin{itemize}
\item Any finite set of $0-1$ random variables that sum to $1$ is
  negatively associated~\cite[Lemma~8]{DubRan98}. So, the set
  $\{X_{i,j}\}_i \cup \{1-\sum_i X_{i,j}\}$ is negatively associated.

  \medskip

\item Any set of increasing functions of pairwise disjoint subsets of
  negatively associated random variables forms, again, a set of
  negatively associated random
  variables~\cite[Proposition~7]{DubRan98}. Thus, each of the sets
  $\{X_{i,j}\}_i$, $\{Y_{i,j}\}_i$, $\{X_i\}_{i \in [s]}$, and
  $\{Y_i\}_{i \in [s]}$ consists of negatively associated random
  variables.
  
  \medskip

\item The Chernoff-Hoeffding bounds apply to sums of any set of
  negatively associated random
  variables~\cite[Proposition~5]{DubRan98}.
\end{itemize}

\noindent
Hence, applying~\cite[Theorem~4.2]{MotRag95} for $\delta=\frac{1}{2}$ for
  instance yields the desired concentration:
  \[ \Prob{X \le \frac{\Ex{X}}2} \le 0.89^{\Ex{X}} \quad \text{and similarly } \quad  \Prob{Y \le \frac{\Ex{Y}}2} \le 0.89^{\Ex{Y}}.\]

\paragraph{Computations.}

(Due to space limitations, some computations are abridged here and
presented in full details in Appendix~\ref{a:L}). Each $B_i$ has area
$c_1/s$, and each $R_i$ has area $c_2/s$ with $c_1 = \frac{\pi}{200}$
and $c_2=\frac{7\pi}{800}$. Thus, $X_{i,j}$ and $Y_{i,j}$ are $0-1$
random variables, taking value $1$ with probability, respectively,
$c_1/s$ and $c_2/s$. For fixed $i$, the $\{X_{i,j}\}_{j \in [n-1]}$
are independent, so we have $\Ex{X_i} = c_1\frac{\log^x n}{n} -
O\pth{\frac{\log^{2x} n}{n^2}}$ and $\Ex{Y_i} \ge c_2\frac{\log^x
n}{n} - O\pth{\frac{\log^{2x} n}{n^2}}$. Since the $X_i$ are
identically distributed, and so are the $Y_i$, we have
\[ \Ex{X} = s \Ex{X_i} \ge c_1n - O\pth{\log^{x} n} \quad \text{and} \quad \Ex{Y} = s \Ex{Y_i} \ge c_2n - O\pth{\log^{x} n}.\]
Now, let $\OO$ denote the event that there exist $a,b$ in $[s]$ such
that each of $B_a$, $R_{a+1}$, $B_b$, $R_{b-1}$ is hit by $P$. Let us
condition by the event $\G = \{X \geq \Ex{X}/2 \hbox{\ and\ } Y \geq
\Ex{Y}/2\}$. A union bound yields
\[ \Prob{\G} \ge 1- \pth{0.89^{\Ex{X}} + 0.89^{\Ex{Y}}} \ge 1 - \pth{0.89^{c_1n - O\pth{\log^{x} n}} + 0.89^{c_2n - O\pth{\log^{x} n}}} \]
which is exponentially close to $1$. We thus bound from below
$\Prob{\OO} \ge \Prob{\G}\Prob{\OO|\G}$ and concentrate on the
conditional probability.

\paragraph{Bichromatic birthday paradox.}

(Due to space limitations, some computations are abridged here and
presented in full details in Appendix~\ref{a:L}). The probability
$\Prob{\OO | \G}$ can be expressed as a convex combination of the
conditional probabilities $f(\beta,\rho)
= \Prob{ \OO|\G_{\beta,\rho}}$, where for integers
$\beta\geq \Ex{X}/2$ and $\rho\geq
\Ex{Y}/2$ we take $\G_{\beta,\rho} = \{X=\beta, Y=\rho\}$. Conditioned on
$\G_{\beta,\rho}$, the occupied regions of each type are uniformly
random and independent, which simplifies the analysis. Furthermore,
the function $f(\beta,\rho)$ is increasing in both variables (the more
occupied regions there are, the more likely it is that the collisions
we desire occur). Thus, we concentrate on finding a lower bound on
$f(\beta,\rho)$ for $\beta=\ceil{ \frac{\Ex{X}}2}$ and
$\rho=\ceil{ \frac{\Ex{Y}}2 }$.

\bigskip

Assume the $\beta$ occupied blue regions have been chosen. Let $T_+$
(resp. $T_-$) denote the set of red regions in sectors following
counterclockwise (resp. clockwise) the sectors whose blue regions have been
chosen. Since the blue regions in the boundary angular sectors may be
among those chosen, we have $\beta-1 \le |T_+|, |T_-| \le \beta$. We now pick
the $\rho$ red regions to be occupied.
Let $E_+$ (resp. $E_-$) denote the event that a region of $T_+$
(resp. $T_-$) has been chosen among the $\rho$ red regions. Pretend, for
the sake of the analysis, that we choose the red regions one by
one. If none of the first $i$ regions chosen is in $T_+$, then next
one has to be picked from the $s-i$ unpicked regions, at least $\beta-1$
of which are in $T_+$. Thus,
\[ 1-\Prob{E_+} \le \prod_{i=0}^{\rho-1} \pth{1-\frac{\beta-1}{s-i}} =   \frac{(s-\rho)!(s-\beta+1)!}{s!(s-\beta-\rho+1)!}\]
Using a symmetric argument for $T_-$ and applying a union bound, we
get $1-f(\beta,\rho) \le
2 \ \frac{(s-\rho)!(s-\beta+1)!}{s!(s-\beta-\rho+1)!}$. Note that for
$\beta=\ceil{ \frac{\Ex{X}}2}$ and $\rho=\ceil{\frac{\Ex{Y}}2}$, both
$\beta$ and $\rho$ are $\Theta(n) = o(s)$. Taking logarithm and using
Simpson's approximation formula, which asserts that $\log(N!) =
N \log(N) - N + O(\log N)$, we get
\[\begin{aligned}
\log(1- f(\beta,\rho)) & = s \log \pth{1+\frac{\rho(\beta-1)}{s(s-\beta-\rho+1)}} - \rho\log \pth{1+\frac{\beta-1}{s-\beta-\rho+1}}\\
&  \hspace{4.8cm} -\beta \log \pth{1+\frac{\rho}{s-\beta-\rho+1}} + O\pth{\log s}.\\
\end{aligned}\]

\noindent
Now in the regime we are looking at, we have $\beta=c_1n/2 - O\pth{\log^{x} n}$,
$\rho=c_2n/2 - O\pth{\log^{x} n}$, and $s=\frac{n^2}{\log^xn}$. 
Taking first order Taylor
expansions, our bound rewrites as

\[ \log(1- f(\beta,\rho)) = -\frac{\rho\beta}{s-\beta-\rho+1} + O(\log s) = -\frac{c_1c_2}{4}\log^x n + O\pth{\log n}.\]

\noindent
provided we have $x>1$. Hence, $f(\beta,\rho) = 1 -
\exp(\Theta(\log(n)^x))$. Altogether, we get that $\Prob{\OO|\G}$ is
exponentially close to $1$. Since $\Prob{\G}$ is also exponentially
close to $1$, we finally get that our event $\OO$ holds with
probability exponentially close to~$1$. With Corollary~\ref{c:anniv},
this proves Lemma~\ref{l:tailL}.

\section{Extension to more general measures}
\label{s:conclusion}

We stated and proved our main result (Theorems~\ref{th:unif}) for a
uniform sample of the unit square. The careful reader may observe,
however, that we have taken care to separate the geometric from the
probabilistic arguments. Although the multiplicative constants of the
leading terms in the end-results matter (we want both $\Ex{U(i)}$ and
$\Ex{L(i)}$ to equal $2 \log n$ at first order), the multiplicative
constants in the geometric arguments do \emph{not} matter:

\begin{itemize}
\item In the analysis of $U$, if Lemma~\ref{l:butt} (in
    appendix) is degraded from $\frac6m + \frac{4}{m \delta(p,q)}$ to
    $O(\frac1{m\delta(p,q)})$, the statement in Lemma~\ref{l:ExU}
    remains that $\Ex{U(1)} \le 2\log n + O(1)$.

  \medskip
  
\item If the blue disk and red annulus are scaled by a constant
  factor, Lemma~\ref{l:cross} still holds with
  $\downarrow_{\frac{\pi}{2s}}$ replaced by
  $\downarrow_{\Theta\pth{\frac{1}{s}}}$; this changes the choice of
  $s$ in Section~\ref{s:Lbb} to $s = \Theta\pth{\frac{n^2}{\log^x
      n}}$, which changes only at \emph{which} exponential speed the
  probability that $L(1) \ge 2\log n - O(\log \log n)$ converges to
  $1$.

  \medskip
  
\item More generally, the lower bound on $L(1)$ should work for any
  probability measure for which one can prove a uniform lower bound of
  $\Omega(1/s)$ for the probabilities of the individual blue and red
  regions.
\end{itemize}

\noindent It should therefore be clear that the same analysis, with
different constants, holds for a variety of more general probability
distributions for the points; examples include the uniform
distribution on any bounded convex domain with non-empty interior, or
even any distribution on such a convex set with a density that is
bounded away from $0$.

\clearpage

\bibliography{./biblio}

\clearpage
\appendix

\section{Experimental setup}\label{a:exp}

We explain here in more detail how we conducted our experiments.

\paragraph{Signature.}

We keep track of the order types already seen by storing them
explicitly in the form of a signature word. Let $P$ be a set of $n$
points. For a labelling $\sigma$ of $P$ by $1, 2, \ldots, n$, let
\[ w(\sigma) = 2a_{1,1} a_{1,2} \ldots a_{1,n-2}1a_{2,1}a_{2,2}
\ldots a_{2,n-2}1a_{3,1}a_{3,2} \ldots a_{3,n-2}\ldots 1a_{n,1}a_{n,2}
\ldots a_{n,n-2} \]
where $1a_{i,1}, a_{i,2}, \ldots$ are the labels of the points in
circular counterclockwise (CCW) order around the $i$th point, starting
from the first point (or from the second point when turning around the
$1$st point). We define as \emph{signature} of $P$ the word
$w(\sigma)$ that is lexicographically smallest among the labelings
$\sigma$ where the points labelled $1$ and $2$ are consecutive on the
convex hull (in CCW order). The fact that this signature characterizes
the order type of $P$ follows from Bokowski's study of hyperline
sequences~\cite[$\mathsection 1.6$]{bokowski2006computational}.

\begin{lemma}\label{signature}
Given a set of permutations of $n$ elements $\{\tau_i\}_{1\leq i\leq
n}$ with $\tau_1(1)=2$ and $\tau_i(1)=1$ for $i>2$ obtained as the
signature of a point set $P$, the orientation of a triple
$(p_a,p_b,p_c)$ with $a<b<c$ depends only on the comparison of
$\tau_a(b)$ and $\tau_a(c)$.
\end{lemma}
\begin{proof}
The ambiguity comes from the fact that when we sort points around
$p_a$ the angle $\widehat{p_bp_ap_c}$ may be greater or smaller than $\pi$.
  Up to an affine transformation, we may assume that for a realization
  of the order type: $p_1$ is the origin,  $p_2=(0,-1)$, and $p_a=(x_a,0)$.
  Then, since $p_1p_2$ is the edge of the convex hull after $p_1$ in
  counter-clockwise direction we get that
  $\forall 2<j<a, x_j>0, y_j<0$ and 
  $\forall j>a, x_j>0, y_j>0$.
  In particular the angle $\widehat{p_bp_ap_c}\leq\pi$
  and we deduce $\tau_a(b)<\tau_a(c) \Longleftrightarrow  {\rm orient}(p_a,p_b,p_c)={\tt ccw}$.
\end{proof}

Actually the above lemma allows to reduce the signature from an
element of $[1,n]^{(n-1)\times n}$ to an element of $[3,n]^{n-2}\times
[3,n]^{n-2}\times [4,n]^{n-3}\times\ldots\times
[i,n]^{n-i+1}\times\times\ldots\times [n-1,n]^2$ reducing
approximatively the signature size by a factor of 2. The signature, as
well as its reduced form, can be computed in time $O(n^3)$ in a
straightforward way. The geometric computations are done using CGAL's
\texttt{Exact\_predicates\_inexact\_constructions\_kernel} \cite{cgal:eb-19a}.)

\paragraph{Pseudorandomess and precision.}

We generated our point sets by picking the coordinates of each point
in $[1.0, 2.0]$. We used the pseudo-random generators of the standard
C++ library, specifically we produce each point's coordinate by a call
to \texttt{dis(gen)} with:

\begin{itemize}
\item[] \texttt{std::random\_device rd;}
\item[] \texttt{std::mt19937\_64 gen(rd());}
\item[] \texttt{std::uniform\_real\_distribution<double> dis(1.0, 2.0);} 
\end{itemize}

\noindent
As a consequence, the precision is the same everywhere in the domain
we sample, and every coordinate is given with $52$ bits of
precision. Note that every order type of size $11$ can be represented
exactly with $16$ bits of precision per
coordinate~\cite{aichholzer2002enumerating}.

\paragraph{Order types of size $10$ and $11$.}

Aichholzer et al.~\cite{aichholzer2002enumerating} enumerated the
order types up to size~$11$. Their count is, however, up to
reflection: they identify the order type of a point set with the order
type of the reflection of that point set with respect to a line. For
size up to~$10$, they also readily provide
realizations.~\footnote{\url{http://www.ist.tugraz.at/staff/aichholzer/research/rp/triangulations/ordertypes/}}
We examined every realization of size~$10$ in their database and
checked whether reflecting the points (horizontally) yields the same
order type; this happened for $13\, 064$ of the realizations. So, the
total number of order types of size $10$ is $28\, 606\, 030$. We
haven't yet done this for size~$11$, so we can only state that their
number is between $2\ 334\ 512\ 907$ and twice that number. If the
small number of symmetric order types of size~$10$ is indicative, we
should expect that the number of order types of size~$11$ to be about
$4.6$ billion.

\bigskip

\newcounter{postponedtheorem:SAVE}
\noindent\begin{minipage}{0.6\textwidth}
\section{Analysis of $U(1)$ for a uniform sample of the unit square\label{a:U}} 
Recall that $P$ is a uniform sample of the unit square of size $n$, 
and $G_m$ is the partition of $[0,1]^2$ into $m\times m$ square cells 
of side length $\frac1m$. 

\subsection{Butterflies}
Given two points $p, q \in [0,1]^2$, we first let $B$ be the union of 
all lines intersecting the cells of $G_m$ containing $p$ and $q$; we 
then define $B_m(p,q)$ as the intersection of 

\end{minipage}~\begin{minipage}{0.4\textwidth}
 \includegraphics[page=1,width=\textwidth]{Figures}
\end{minipage}\smallskip 

\noindent 
$[0,1]^2$ with the 
Minkowski sum of $B$ with a disk of radius {$\frac{\sqrt{2}}m$}.  We 
call $B_m(p,q)$ the \emph{butterfly} of $p$ and $q$ (at resolution 
$m$). Note that the butterfly $B_m(p,q)$ contains all the cells 
intersecting $B$. Hence, if there exists a line intersecting the cells 
of $p$, $q$ and $r$, then $r\in B_m(p,q)$. The following lemma bounds 
the area of $B_m(p,q)$ by $O(\frac1{m\delta(p,q)})$.

\begin{lemma}\label{l:butt}
  The area of $B_m(p,q)$ is at most $\frac6m + \frac{4}{m
    \delta(p,q)}$ where $\delta(p,q)$ is the distance between the
  centers of the cells of $p$ and $q$.
\end{lemma}
\begin{proof}
  This lemma is similar to Lemma~3 in Fabila-Monroy and Huemer paper
  \cite{fabila2017order},
  although their analogous of butterfly have different definition.
  Note that the bound holds trivially if $p$ and $q$ are in the same 
  cell ($\delta(p,q)=0$) or in adjacent cells 
  ($\delta(p,q)=1/m$). Otherwise, the butterfly $B_m(p,q)$ consists of 
  two parts: a strip $S_m(p,q)$ and the union $T_m(p,q)$ of four 
  triangles (shaded in, respectively, green and blue in 
the figure). 
 We have 

\noindent\begin{minipage}{0.4\textwidth}
 \includegraphics[page=2,width=\textwidth]{Figures}
\end{minipage}~\begin{minipage}{0.6\textwidth}

\[ \area\pth{S_m(p,q)} \le \pth{3\frac{\sqrt{2}}m} \cdot \sqrt{2} = \frac6m.\]

The four triangles come in two pairs of homothetic 
triangles, intersected with $[0,1]^2$. Each homothetic pair 
consists of images under scaling of a triangle whose basis is the half 
diagonal of a cell of length $\frac{\sqrt{2}}{2m}$ and whose height 
$h$ is at least $\frac{\delta(p,q)}{2\sqrt{2}}$
(the two kinds of triangles have blue and red boundaries in 
the figure). 
Letting $u$ and $v$ denote the scaling 
\end{minipage}
 factors, the areas of the two homothetic triangles sum to 
$\frac12(u^2+v^2)h\frac{\sqrt{2}}{2m}$.  Since the scalings turn the 
height of the reference triangles to two lengths that sum 
to\footnote{The heights are smaller than the sides and the sides are 
  inside the square $[0,1]^2$ and have disjoint projection on the line 
  $(pq)$.} at most $\sqrt{2}$, we have $(u+v)h \le\sqrt{2}$.  This 
implies that $u^2+v^2 \le\pth{\frac{\sqrt{2}}h}^2 $ and one pair of 
homothetic triangles contributes at most $\frac12 
\frac2{h^2}h\frac{\sqrt2}{2m} =\frac{\sqrt{2}}{2hm} \le 
\frac{2}{m\delta(p,q)} $. Altogether, $\area\pth{T_m(p,q)}
\le\frac{4}{m\delta(p,q)}$.  Finally $\area\pth{B_m(p,q)} \le\frac6m +
\frac{4}{m\delta(p,q)}.$
\end{proof}

\subsection{Distribution of $U(1)$}

We now analyze the distribution function of the random variable 
$U(1)$. Recall that the randomness here refers to the choice of the 
random points $p_1, p_2, \ldots p_n$, which are taken independently 
and uniformly in $[0,1]^2$.

\begin{lemma}\label{l:dU}
$\Prob{U(1) > k}  \le \ljp n^2 2^{-k}$.
\end{lemma}
\begin{proof}
  This lemma is similar to Lemma~4 in Fabila-Monroy and Huemer paper \cite{fabila2017order}. 
  In our setting, we have: 
  \begin{eqnarray*}
    \Prob{U(1) > k}  & \leq &  \Prob{\exists i,j \in \binom{[n]\setminus \{1\}}2 
      : p_j \in B_{2^{k}}(p_1,p_i)}\\
    & \le & (n-1) \Prob{\exists j \in [n]\setminus \{1,2\} : p_j \in B_{2^{k}}(p_1,p_2)}\\
    & \le & (n-1) \Ex{1-(1-\area\pth{B_{2^{k}}(p_1,p_2)})^{n-2}}. 
  \end{eqnarray*}
  The geometry of $B_{2^k}(p_1,p_2)$ depends on the distance between 
  the centers of the cells that contain $p_1$ and $p_2$. We therefore 
  condition on the cell containing $p_1$, then sum the contributions 
  of the cell containing $p_2$ by distance to the cell containing 
  $p_1$.  Accounting for boundary effects, for any $1 \le t \le 2^k$
  there are at most $8t$ cells whose center lies at a distance between 
  $t2^{-k}$ and $(t+1)2^{-k}$ from a given cell. We thus have (using Lemma~\ref{l:butt})
  \begin{eqnarray*}
    \lefteqn{ \Ex{1-(1-\area B_{2^k}(p_1,p_2))^{n-2}} }\\
    & = &\sum_{c \in \text{ cells of } G_{2^k}} \Prob{p_2\in c} \cdot \Ex{1-(1-\area B_{2^k}(p_1,p_2))^{n-2}\mid p_2\in c} \\
    & \le & \sum_{t=1}^{2^k} \frac{8t}{(2^k)^2} \pth{1-\pth{1-\pth{\frac{6}{2^{k}}+\frac{4}{2^k\cdot(t2^{-k})}}}^{n-2}}. 
  \end{eqnarray*}
  Using $(1-x)^{n-2} \ge 1-(n-2)x$ we get 
  \begin{eqnarray*}
    \Ex{1-(1-\area B_{2^k}(p_1,p_2))^{n-2}} & \le &(n-2) 2^{-2k}
    \sum_{t=1}^{2^k} 8 t \pth{6\cdot2^{-k}+\frac{4}t} \\
    & \le &n2^{-2k} \pth{\sum_{t=1}^{2^k} 48 t2^{-k}} + n 2^{-2k} 2^k 32\\
    & \le &24n2^{-k}(1+2^{-k})+ 32n2^{-k}. 
  \end{eqnarray*}
  The statement trivially bounds a probability by something greater than 1 for  $k\leq 5$. For $k \ge 6$, the final term is at most $\ljp n2^{-k}$. 
\end{proof}

\setcounter{postponedtheorem:SAVE}{\value{theorem}}
\setcounter{theorem}{\value{postponedtheorem:ExU}}
\begin{lemma}
  $\Ex{U(1)} \le 2\log n + 8$
\end{lemma}
\setcounter{theorem}{\value{postponedtheorem:SAVE}}
\begin{proof}
  By definition we have 
  \[ \Ex{U(1)} = \sum_{k=1}^{\infty} k \Prob{U(1)=k} = \sum_{k=0}^{\infty}
  \Prob{U(1)>k} .\]
  For the first $2 \log n + 6$ terms, we use the trivial upper bound of $1$ and for the remaining terms we use the upper bound of Lemma~\ref{l:dU}: 
  \[ \Ex{U(1)} \le (2 \log n+6) + \ljp n^2 \sum_{k\ge 6+2 \log n}^{\infty} 2^{-k} = (2 \log n + 6) +  \ljp n^2\cdot2^{-5-2 \log n} .\]
  Altogether it comes that $\Ex{U(1)} \le 2 \log n + 8$. 
\end{proof}

\section{Detailed proofs for Section~\ref{s:L}\label{a:L}}

\begin{minipage}{0.8\textwidth}
\setcounter{postponedtheorem:SAVE}{\value{theorem}}
\setcounter{theorem}{\value{postponedtheorem:cross}}
\begin{lemma}
  If $s\ge 10$, for any point $b \in B_i$ and $r \in R_{i+1}$, the line
  $(br)$ intersects $\downarrow_{\frac{\pi}{2s}}$.
\end{lemma}
\setcounter{theorem}{\value{postponedtheorem:SAVE}}
\end{minipage}
\begin{proof}
  The vertical distance between $p_1$ and $(br)$ is maximal when $b$ and\\
\begin{minipage}{0.8\textwidth}
 $r$ are placed 
  in the corners of $B_{s-1}$ and $R_{s}$ on 
  circles of radii 0.2 and 0.3 as in left figure. 
Let  us relate this maximal distance $h$ to $\theta =
\widehat{bp_1r}$. 
With the notations of the right figure, 
considering triangle $vp_1b$ we have 
$\beta+\gamma+(\frac{7\pi}{8}-\theta)=\pi$ and deduce 
$\theta=\beta+\gamma-\frac\pi8$. 
Law of sines in the same triangle give 
$\sin\beta=\frac{pr}{vr}\sin \frac{7\pi}{8}
= \frac{0.3}{\sqrt{0.09+h^2+0.6h\sin\frac\pi8}}\sin\frac\pi8$
and 
$\sin\gamma=\frac h{pb}\sin\beta=\frac h{0.2}\sin\beta 
$. And 
we  can express $\theta$ as a function of $h$ (for $\theta$ sufficiently 
  small): 

\end{minipage}\quad\begin{minipage}{0.2\textwidth}
\vspace*{-2.5cm}\includegraphics[page=6,width=\textwidth]{Figures}
\end{minipage}

\[
\theta = \arcsin\pth{ \tfrac{0.3}{\sqrt{0.09+h^2+0.6h\sin\tfrac\pi8}}
  \sin\tfrac\pi8 } + \arcsin\pth{
  \tfrac{0.3}{\sqrt{0.09+h^2+0.6h\sin\tfrac\pi8}} \tfrac h{0.2}
  \sin\tfrac\pi8 } -\tfrac\pi8. 
\]

\noindent 
This function $h \mapsto \theta(h)$ is increasing on $[0,0.6]$ and 
$\theta(h) > h$ when $\theta(h) \in [0,0.17]$. Since $\theta$ is the 
angle of \emph{two} sectors, we have $\theta =2\frac{\pi}{4s}$. For $s 
\ge 10$ we have $h<\frac{\pi}{2s}$. 
\end{proof}

\paragraph{Computations.}

Each $B_i$ has area $c_1/s$, and each $R_i$ has area $c_2/s$ with $c_1
= \frac{\pi}{200}$ and $c_2=\frac{7\pi}{800}$. Thus, $X_{i,j}$ and
$Y_{i,j}$ are $0-1$ random variables, taking value $1$ with
probability, respectively, $c_1/s$ and $c_2/s$. For fixed $i$, the
$\{X_{i,j}\}_{j \in [n-1]}$ are independent, so we have

\[\begin{aligned}
\Ex{X_i} = \Prob{X_i = 1} = 1 - \pth{1-\frac{c_1}{s}}^{n-1} \ge 1-e^{-c_1\frac{n-1}{s}} & \ge c_1\frac{n-1}{s} - \frac12 \pth{ c_1\frac{n-1}{s}}^2 \\
&= c_1\frac{\log^x n}{n} - O\pth{\frac{\log^{2x} n}{n^2}}.
\end{aligned}\]

\noindent
the first and second inequalities coming, respectively, from the facts
 that for every $t \ge 0$ we have $1-t \le e^{-t}$ and for every $t \in
[0,1]$ we have $1-e^{-t} \ge -t-\frac{t^2}2$.
Then, we plugged in $s=\frac{n^2}{\log^x n}$.
The same computation gives
$ \Ex{Y_i} \ge c_2\frac{\log^x n}{n} - O\pth{\frac{\log^{2x} n}{n^2}}.$
\noindent
Finally, since the $X_i$ are identically distributed, and so are the
$Y_i$, we have

\[ \Ex{X} = s \Ex{X_i} \ge c_1n - O\pth{\log^{x} n} \quad \text{and} \quad \Ex{Y} = s \Ex{Y_i} \ge c_2n - O\pth{\log^{x} n}.\]

\noindent
Now, let $\OO$ denote the event that there exist $a,b$ in $[s]$ such
that each of $B_a$, $R_{a+1}$, $B_b$, $R_{b-1}$ is hit by $P$. Let us
condition by the event $\G = \{X \geq \Ex{X}/2 \hbox{\ and\ } Y \geq
\Ex{Y}/2\}$. A union bound yields

\[ \Prob{\G} \ge 1- \pth{0.89^{\Ex{X}} + 0.89^{\Ex{Y}}} \ge 1 - \pth{0.89^{c_1n - O\pth{\log^{x} n}} + 0.89^{c_2n - O\pth{\log^{x} n}}} \]

\noindent
which is exponentially close to $1$. We thus bound from below
$\Prob{\OO} \ge \Prob{\G}\Prob{\OO|\G}$ and concentrate on the
conditional probability.

\paragraph{Bichromatic birthday paradox.}

The probability $\Prob{\OO | \G}$ can be expressed as a convex
combination of the conditional probabilities $f(\beta,\rho) = \Prob{
  \OO|\G_{\beta,\rho}}$, where for integers $\beta\geq \Ex{X}/2$ and $\rho\geq
\Ex{Y}/2$ we take $\G_{\beta,\rho} = \{X=\beta, Y=\rho\}$. Conditioned on
$\G_{\beta,\rho}$, the occupied regions of each type are uniformly
random and independent, which will simplify the analysis. Furthermore,
the function $f(\beta,\rho)$ is increasing in both variables (the more
occupied regions there are, the more likely it is that the collisions
we desire occur). Thus, we concentrate on finding a lower bound on
$f(\beta,\rho)$ for $\beta=\ceil{ \frac{\Ex{X}}2}$ and
$\rho=\ceil{ \frac{\Ex{Y}}2 }$.

\bigskip

Assume the $\beta$ occupied blue regions have been chosen. Let $T_+$
(resp. $T_-$) denote the set of red regions in sectors following
counterclockwise (resp. clockwise) the sectors whose blue regions have been
chosen. Since the blue regions in the boundary angular sectors may be
among those chosen, we have $\beta-1 \le |T_+|, |T_-| \le \beta$. We now pick
the $\rho$ red regions to be occupied.
Let $E_+$ (resp. $E_-$) denote the event that a region of $T_+$
(resp. $T_-$) has been chosen among the $\rho$ red regions. Pretend, for
the sake of the analysis, that we choose the red regions one by
one. If none of the first $i$ regions chosen is in $T_+$, then next
one has to be picked from the $s-i$ unpicked regions, at least $\beta-1$
of which are in $T_+$. Thus,

\[ \begin{aligned}
  1-\Prob{E_+} \le \prod_{i=0}^{\rho-1} \pth{1-\frac{\beta-1}{s-i}} & = \frac{(s-\beta+1)(s-\beta)\ldots (s-\beta-\rho+2)}{s(s-1) \ldots (s-\rho+1)}\\
  & =   \frac{(s-\rho)!(s-\beta+1)!}{s!(s-\beta-\rho+1)!}\\
  \end{aligned}\]

\noindent
Using a symmetric argument for $T_-$ and applying a union bound, we get
\[ 1-f(\beta,\rho) \le 2 \ \frac{(s-\rho)!(s-\beta+1)!}{s!(s-\beta-\rho+1)!}.\]
\noindent
Note that for $\beta=\ceil{ \frac{\Ex{X}}2}$ and
$\rho=\ceil{\frac{\Ex{Y}}2}$,
 both $\beta$ and $\rho$ are $\Theta(n) = o(s)$. Taking
logarithm and using Simpson's approximation formula, which asserts
that $\log(N!) = N \log(N) - N + O(\log N)$, we get

\[\begin{aligned}
\log(1- f(\beta,\rho)) & = s \log \frac{(s-\rho)(s-\beta+1)}{s(s-\beta-\rho+1)} - \rho\log \frac{s-\rho}{s-\beta-\rho+1}\\
&  \hspace{4cm} -\beta \log \frac{s-\beta+1}{s-\beta-\rho+1} + O\pth{\log s}\\
& = s \log \pth{1+\frac{\rho(\beta-1)}{s(s-\beta-\rho+1)}} - \rho\log \pth{1+\frac{\beta-1}{s-\beta-\rho+1}}\\
&  \hspace{4.8cm} -\beta \log \pth{1+\frac{\rho}{s-\beta-\rho+1}} + O\pth{\log s}.\\
\end{aligned}\]

\noindent
Now in the regime we are looking at, we have $\beta=c_1n/2 - O\pth{\log^{x} n}$,
$\rho=c_2n/2 - O\pth{\log^{x} n}$, and $s=\frac{n^2}{\log^xn}$. 
Taking first order Taylor
expansions, our bound rewrites as

\[ \log(1- f(\beta,\rho)) = -\frac{\rho\beta}{s-\beta-\rho+1} + O(\log s) = -\frac{c_1c_2}{4}\log^x n + O\pth{\log n}.\]

\noindent
provided we have $x>1$. Hence, $f(\beta,\rho) = 1 -
\exp(\Theta(\log(n)^x))$. Altogether, we get that $\Prob{\OO|\G}$ is
exponentially close to $1$. Since $\Prob{\G}$ is also exponentially
close to $1$, we finally get that our event $\OO$ holds with
probability exponentially close to~$1$. With Corollary~\ref{c:anniv},
this proves Lemma~\ref{l:tailL}.

\end{document}